\documentclass{svmult}

\usepackage{mathptmx}       
\usepackage{helvet}         
\usepackage{courier}        
\usepackage{type1cm}       
                            
\usepackage{makeidx}         
\usepackage{graphicx}        
                            
\usepackage{multicol}        
\usepackage[bottom]{footmisc}
\usepackage{amssymb}
\usepackage{amsfonts}
\usepackage{amsmath}
\usepackage{hyperref}
\usepackage{multirow}
\usepackage{bigstrut}

\renewcommand{\rightrightarrows}{:=}

\begin{document}

\title*{Discrete-Time Quadratic Hedging of Barrier Options in Exponential L\'{e}vy Model}

\author{Ale\v{s} \v{C}ern\'{y}}
\institute{Ale\v{s} \v{C}ern\'{y} \at Cass Business School, City University London, 106 Bunhill Row,
EC1Y 8TZ London, UK, \email{Ales.Cerny.1@city.ac.uk}}

\maketitle

\abstract{We examine optimal quadratic hedging of barrier options in a discretely
sampled exponential L\'{e}vy model that has been realistically calibrated to
reflect the leptokurtic nature of equity returns. Our main finding is that
the impact of hedging errors on prices is several times higher than the
impact of other pricing biases studied in the literature.}

\section{Introduction}

We study quadratic hedging and pricing of European barrier options with a
particular focus on the magnitude of risk of optimal hedging strategies. In
a discretely sampled exponential L\'{e}vy model, calibrated to reflect the
leptokurtic nature of equity returns, we compute the hedging error of the
optimal strategy and evaluate prices that yield reasonable\ risk-adjusted
performance for the hedger. We also confirm what traders already know
empirically, namely that the hedging risk of barrier options substantially
outstrips that of plain vanilla options.

European barrier options are derivative contracts based on standard European
calls or puts with the exception that the option becomes active (or
inactive) when the stock price hits a prespecified barrier before the
maturity of the option. Options activated in this way are called knock-ins;
those deactivated are called knock-outs.

Under the assumptions of the Black--Scholes model barrier options have been
valued first by \cite{Mer73b} and in more detail by \cite{ReiRub91}. Early
literature on numerical evaluation of barrier option prices concentrates on
slow convergence of binomial method, which is due to the difference between
the \emph{nominal barrier} specified in the option contract and the \emph{%
effective barrier }implied by the position of nodes in the stock price
lattice. This discrepancy, if not properly controlled, may lead to sizeable
mispricing, especially for options whose barrier is close to the initial
stock price. \cite{BoyLau94} and \cite{Rit95} suggest better positioning
of nodes in binomial and trinomial lattices to minimize the discrepancy
between nominal and effective barrier, whereas \cite{Deretal95}
propose interpolation between two adjacent values of the effective barrier. %
\cite{FigGao99} devise an adaptive mesh allowing for more nodes (and
shorter time steps) around the barrier.

The papers above are concerned with continuously-monitored barriers in the
Black-Scholes model. Discrete monitoring, too, can have significant impact
on option valuation, and, unlike the continuous monitoring case, does not
allow for a simple closed form pricing formula, cf. \cite{fusai.al.06}. 
A simple asymptotic correction, which works well for barriers not too
close to initial stock price, was developed by \cite{Broetal97}. For
barriers in close proximity of the stock price the Markov chain
representation of stock prices developed by \cite{Duaetal03} is more
appropriate. Other papers dealing with discrete monitoring in the log-normal
framework include \cite{fusai.recchioni.07}, \cite{Hor03a}, %
\cite{kou.03}, \cite{kuan.webber.04} and \cite{wei.98}. %
\cite{andricopoulos.al.03} describe a systematic way of handling
discretization errors by means of quadrature. There is also an extensive
literature on barrier option pricing by Monte Carlo simulation which we will
not touch upon in this paper.

The models discussed above are complete in the sense that one can devise a
self-financing trading strategy that perfectly replicates the barrier
option. In practice, however, one encounters considerable difficulties in
maintaining a delta-neutral position when close to the barrier. This has
motivated study of static replication of barrier options with plain vanilla
options. \cite{Caretal98} use the reflection principle known from
barrier option pricing combined with so-called put--call symmetry to write a
down-and-out call as a sum of a long call and a short put. Their methodology
is to some extent model-free but it only works if the market is complete and
if the aforementioned symmetry holds, requiring that risk-free rate be equal
to the dividend rate in addition to a certain symmetry of local
volatilities. \cite{Broetal01} analyze static super- and
sub-replication. The latter results are completely model-free at the cost of
generating price bounds that are potentially very wide. Other papers on
static replication include \cite{bowie.carr.94}, %
\cite{carr.chou.97}, \cite{derman.et.al.94} and %
\cite{derman.et.al.95}.

Several studies allow for parametric departures from the Black--Scholes
model. \cite{Dup01} and \cite{IlhSir06} use Bates' stochastic volatility
jump-diffusion model while \cite{Lei99} allows for IID jumps. Several
numerical approaches now exist for dealing with a wide class of (possibly
infinite variation) L\'{e}vy models, see \cite{feng.linetsky.08}, %
\cite{fang.oosterlee.09}.

The paper is organized as follows. In Section 2 we specify the theoretical
model, describe its calibration and computation of optimal strategies.
Section 3 provides economic analysis of the numerical results and Section 4
explains the relationship between barrier option prices and hedger's
risk-adjusted performance.

\section{The model}

We have at our disposal nominal log returns on FT100 equity index in the
period January 1st, 1993 to December 31st, 2002, sampled at a 1 minute
interval. Eventually we wish to say something about optimal hedging of
barrier options in a model with rebalancing frequency $\Delta \in \lbrack 5$
minutes, $1$ day$]$ and a daily monitoring of the barrier. We will assume
independence and time homogoneity of underlying asset returns at any given
rebalancing frequency. This is not to say that stochastic volatility is
unimportant in practice, instead we may think of the IID assumption as a
useful limiting case when the (unobserved) volatility state changes either
very slowly or very quickly. In this view of the world the leptokurtic
nature of returns is a source of risk that does not vanish even after
stochastic volatility has been factored in appropriately.

The analysis is performed under two self-imposed constraints. The first is
to use the available data in a non-parametric way and the second is to
perform all numerical analysis in a multinomial lattice.

In these circumstances there are essentially two strategies for calibrating
the stock price process. One option is to simply take the data series
sampled at time interval $\Delta ,$ generate a discretized distribution of
returns and construct a multinomial lattice using this distribution. An
alternative is to consider an underlying continuous-time model from which
the daily or hourly returns are extracted. \cite{EbeOzk03} argue that
equity return data display sufficient amount of time consistency for such an
approach to make sense. The underlying model is then necessarily a geometric
L\'{e}vy model, cf. Lemma 4.1 in \cite{cerny.07}. Such approach also
offers an alternative avenue to obtaining asymptotics as $\Delta $ tends to
zero by studying quadratic hedging for barrier options directly in the
underlying L\'{e}vy model -- a task which at present is still outstanding
and well beyond the scope of this paper.

\subsection{Calibration\label{Sect: Calibration}}

We take the original log return data sampled at $\Delta _{0}=$ 1 minute
intervals and construct an equidistantly spaced sequence $m_{0}<m_{1}<\ldots
<m_{N+1}$ with spacing $\delta $, such that $m_{N}$ is the highest and $m_{1}
$ the lowest log return in the sample. We set $N=1000$. We then identify the
frequency of log returns in each interval of length $\delta $ centred on $%
m_{j}$, $j=0,\ldots ,N+1$ and store this information in the vector $%
\{f_{j}\}_{j=0}^{N+1}.$ We construct an empirical L\'{e}vy measure $F_{%
\mathrm{raw}}$ as an absolutely continuous measure with respect to the
Lebesgue measure on $\mathbb{R}$%
\begin{equation*}
F_{\mathrm{raw}}\left( dx\right) =\frac{\hat{f}(x)}{\Delta _{0}}dx,
\end{equation*}%
where $\hat{f}=f$ at the points $m_{j},$ $\hat{f}=0$ outside $(m_{0},m_{N+1})
$ and elsewhere $\hat{f}$ is obtained as a linear interpolation of $f$. This
construction is motivated by an asymptotic result that links transition
probability measure of a L\'{e}vy process to its L\'{e}vy measure over short
time horizons, see \cite{Sat99}, Corollary 8.9.

In the next step we normalize the empirical characteristic function of log
returns to achieve a pre-specified annualized mean $\mu $ and volatility $%
\sigma $. Since the raw empirical L\'{e}vy process is square-integrable and
therefore a special semimartingale we will use the (otherwise forbidden)\
truncation function $h(x)=x$. We will construct the log return process by
setting 
\begin{eqnarray*}
\ln S &=&\ln S_{0}+\mu t+\frac{\sigma x}{\sigma _{\mathrm{raw}}}\ast (J^{L_{%
\mathrm{raw}}}-\nu ^{L_{\mathrm{raw}}}), \\
\sigma _{\mathrm{raw}}^{2} &=&\int_{\mathbb{R}}x^{2}F_{\mathrm{raw}}\left(
dx\right) ,
\end{eqnarray*}%
where $J^{L_{\mathrm{raw}}}$ is the jump measure of a L\'{e}vy process with L%
\'{e}vy measure $F_{\mathrm{raw}}$, $\nu ^{L_{\mathrm{raw}}}$ is its
predictable compensator and $\ast $ denotes a certain stochastic integral as
defined in \cite{js.03}. II.1.27. This yields{%
\begin{eqnarray}
\kappa (u):= &&\mu u+\int_{\mathbb{R}}\left( \mathrm{e}^{ux}-1-ux\right) F(%
\mathrm{d}x),  \label{eq: cumulant} \\
F\left( G\right) := &&\int_{\mathbb{R}}1_{G}\left( \frac{\sigma x}{\sigma _{%
\mathrm{raw}}}\right) F_{\mathrm{raw}}\left( dx\right) .  \label{eq: F}
\end{eqnarray}%
} We fix the annualized volatility of log returns at $\sigma =0.2,$ but to
check the robustness of our results we allow the mean log return to take 2
different values $\mu \in \{-0.1,0.1\}$, the first representing a bear
market and the second representing a bull market.

Instead of the non-parametric calibration procedure above one could instead
estimate a model from a convenient parametric family, such as the
generalized hyperbolic family, as outlined in \cite{EbePra02}. The
parametric route offers in some special cases an explicit expression for the
log return density at all time horizons which avoids the need for numerical
inversion of the characteristic function employed below.

\subsection{Multinomial lattice}

If $Z$ denotes the log return on time horizon $\Delta $ its characteristic
function is of the form 
\begin{equation*}
\mathrm{E}\left[ \mathrm{\exp }\left( \mathrm{i}vZ\right) \right] =\mathrm{e}%
^{\kappa (\mathrm{i}v)\Delta },
\end{equation*}%
where the cumulant generating function $\kappa $ is given by equations (\ref%
{eq: cumulant}) and (\ref{eq: F}). Provided that $Z$ has no atom at $z$ the
cumulative distribution is given by the inverse Fourier formula, see \cite{Sat99}, 2.5xi,%
\begin{equation*}
P\left( Z\leq z\right) =\mathcal{H}(c)-\frac{1}{2\pi }\lim_{l\rightarrow
\infty }\int_{-l}^{l}\frac{\mathrm{e}^{\kappa (\mathrm{i}\lambda -c)\Delta
-z\left( \mathrm{i}\lambda -c\right) }}{\mathrm{i}\lambda -c}d\lambda ,
\end{equation*}%
where $c$ is an arbitrarily chosen real number\footnote{%
In numerical calculations with a fixed value of $z$ we choose $c$ so as to
minimize the value of the integrand at $\lambda =0$, see \cite{strawderman.04}, equation (3).} 
and $\mathcal{H}$ is a step function,%
\begin{equation*}
\mathcal{H}\left( x\right) =\left\{ 
\begin{array}{cc}
0 & \text{for }x>0, \\ 
\frac{1}{2} & \text{for }x=0, \\ 
1 & \text{for }x<0,%
\end{array}%
\right. .
\end{equation*}

We now define a discretized distribution of log returns to populate our
lattice. The discretized random variable $\hat{Z}$ will take values%
\begin{equation*}
\hat{z}_{j}=j\eta \text{ with }j\in \lbrack -n_{\mathrm{down}},n_{\mathrm{up}%
}]\cap \mathbb{Z}\text{,}
\end{equation*}%
where $n_{\mathrm{down}},n_{\mathrm{up}}$ are the smallest numbers in $%
\mathbb{N}$ such that $P(Z\leq -n_{\mathrm{down}}\eta )\leq \alpha $ and $%
F_{Z}(Z\leq n_{\mathrm{up}}\eta )\geq 1-\alpha ,$ respectively. We use the
values $\eta =0.0005$ and $\alpha =10^{-5}.$ For comparison, the
corresponding value of $\eta $ in \cite{Duaetal03} with 1001 price
nodes is $0.0089$. Table \ref{Table1} shows the number of standard
deviations.

The transition probabilities corresponding to different values of log return
are defined by%
\begin{eqnarray*}
\hat{p}_{j} &\rightrightarrows &P(Z\leq \left( j+1/2\right) \eta )-P(Z\leq
\left( j-1/2\right) \eta )\text{ for }j\in \left( -n_{\mathrm{down}},n_{%
\mathrm{up}}\right) \cap \mathbb{Z}, \\
\hat{p}_{j} &\rightrightarrows &P(Z\leq \left( j+1/2\right) \eta )\text{ for 
}j=-n_{\mathrm{down}}, \\
\hat{p}_{j} &\rightrightarrows &1-P(Z\leq \left( j-1/2\right) \eta )\text{
for }j=n_{\mathrm{up}}.
\end{eqnarray*}

To limit the effect of the discretization errors arising from an arbitrary
position of the barrier we limit computations to barrier levels that satisfy 
$\ln B-\ln S\in (\mathbb{Z}+1/2)\eta $ and use interpolation otherwise.

\subsection{Optimal hedging}

In the multinomial lattice constructed above we compute the optimal hedging
strategy and the minimal hedging error according to the following theorem.

\begin{theorem}
Suppose that there is an $\mathcal{F}_{n}$-measurable contingent claim $H$
such that $\mathrm{E}[H^{2}]<\infty $. In the absence of transaction costs
the dynamically optimal hedging strategy $\varphi ${\ solving%
\begin{equation*}
\inf_{\vartheta }\mathrm{E}[\left( G_{n}^{x,\vartheta }-H\right) ^{2}],
\end{equation*}%
subject to $\vartheta _{i}$ being }$\mathcal{F}${$_{i}$-measurable with }$G$
being the value of a self-financing portfolio,{%
\begin{eqnarray*}
G_{i}^{x,\vartheta } &=&RG_{i-1}^{x,\vartheta }+\vartheta
_{i-1}(S_{i}-RS_{i-1}), \\
G_{0}^{x,\vartheta } &=&x,
\end{eqnarray*}%
is given by%
\begin{eqnarray}
\varphi _{i} &=&\xi _{i}+aR\frac{V_{i}-G_{i}^{x,\varphi }}{S_{i}},  \notag \\
V_{i} &\rightrightarrows &\mathrm{E}_{i}[(1-aX_{i+1})V_{i+1}]/(bR),
\label{eq: V_i} \\
V_{n} &\rightrightarrows &H  \label{eq: V_n} \\
\xi _{i} &\rightrightarrows &\mathrm{Cov}_{i}\left( V_{i+1},S_{i+1}\right) /%
\mathrm{Var}_{i}\left( S_{i+1}\right)   \notag \\
&=&\mathrm{E}_{i}\left[ \left( V_{i+1}-RV_{i}\right) X_{i+1}\right] /\left(
S_{i}\mathrm{E}_{i}\left[ X_{i+1}^{2}\right] \right) ,  \notag \\
X_{i} &\rightrightarrows &\exp (Z_{i})-R,  \notag \\
a &\rightrightarrows &\mathrm{E}_{i}\left[ X_{i+1}\right] /\mathrm{E}_{i}%
\left[ X_{i+1}^{2}\right] ,  \notag \\
b &\rightrightarrows &1-\left( \mathrm{E}_{i}\left[ X_{i+1}\right] \right)
^{2}/\mathrm{E}_{i}\left[ X_{i+1}^{2}\right] .  \notag
\end{eqnarray}%
The hedging performance of the dynamically optimal strategy }$\varphi $ {and
of the locally optimal strategy }$\xi ${\ is given by%
\begin{eqnarray}
\mathrm{E}\left[ \left( G_{n}^{x,\varphi }-H\right) ^{2}\right]  &=&\left(
R^{2}b\right) ^{n}\left( x-V_{0}\right) ^{2}+\varepsilon _{0}^{2}(\varphi ),
\notag \\
\mathrm{E}\left[ \left( G_{n}^{x,\xi }-H\right) ^{2}\right]  &=&\left(
R^{2}\right) ^{n-j}\left( x-V_{0}\right) ^{2}+\varepsilon _{0}^{2}(\xi ), 
\notag \\
\varepsilon _{0}^{2}(\varphi ) &=&\sum_{j=0}^{n-1}\left( R^{2}b\right)
^{n-j-1}\mathrm{E}\left[ \psi _{j}\right] ,  \label{eq: eps_0} \\
\varepsilon _{0}^{2}(\xi ) &=&\sum_{j=0}^{n-1}R^{2(n-j-1)}\mathrm{E}\left[
\psi _{j}\right] ,  \notag \\
\psi _{j} &\rightrightarrows &\mathrm{E}_{j}\left[ \left( RV_{j}+\xi
_{j}S_{j}X_{j+1}-V_{j+1}\right) ^{2}\right]   \notag \\
&=&\mathrm{Var}_{j}\left( V_{j+1}^{2}\right) -\frac{\left( \mathrm{Cov}%
_{j}\left( S_{j+1},V_{j+1}\right) \right) ^{2}}{\mathrm{Var}_{j}\left(
S_{j+1}\right) }.  \notag
\end{eqnarray}%
}
\end{theorem}

\begin{proof}
See \cite{cerny.07}, Theorem 3.3. \qed
\end{proof}

\section{Numerical results}

We first fix the rebalancing period to $\Delta =1$ day and examine the
behaviour of hedging errors across maturities, strikes and barrier levels.We
then analyze the asymptotics of the hedging error as the rebalancing
interval $\Delta $ approaches $0,$ keeping the monitoring frequency of the
barrier constant. We do so initially for a range of strikes and barrier
levels with rebalancing interval $\Delta =1$ hour and then with fixed strike
and barrier level we examine asymptotics going down to $\Delta =5$ minutes. 

\subsection{Effect of barrier position, maturity and drift\label{Sect: UpOut}%
}

We consider an up-and-out European call and two maturity dates: 1 and 6
months. During a detailed preliminary analysis we have found that changes in
risk-free rate have a very small impact upon hedging errors and therefore we
fix the risk-free rate in all computations to $r=0$. Volatility is
normalized to $\sigma =0.2$ as explained in Section~\ref{Sect: Calibration},
while the drift takes two values $\mu \in \{-0.1,0.1\}.$ The time units
reflect trading time; specifically we assume there are 8 hours in a day and
250 days in a year. To be able to compare the size of hedging error across
maturities we measure the position of the barrier and of the striking price
relative to the initial stock price in terms of their Black--Scholes delta.

For each set of parameters we report five quantities: i)\ the Black--Scholes
price of a continuously monitored option $C$, ii) the Black--Scholes price
of a daily monitored option\footnote{%
Computation of the discretely monitored option price in Black--Scholes model
follows the methodology of \cite{Duaetal03}. Effectively, the
calculation is the same as for $V$ in the empirical model, but the
multinomial transition probabilities approximate the Black-Scholes
risk-neutral distribution $N\left( \left( r-\sigma ^{2}/2\right) \Delta
,\sigma ^{2}\Delta \right) $.} $\hat{V}$, iii)\ the standard deviation of
the hedging error in a discretely rebalanced Black-Scholes model $\hat{%
\varepsilon}_{0}$ obtained from (\ref{eq: eps_0}) using multinomial
approximation of Black-Scholes normal transition probabilities\footnote{%
Objective probability distribution of log returns in the Black-Scholes model
is $N\left( \left( \mu -\sigma ^{2}/2\right) \Delta ,\sigma ^{2}\Delta
\right) .$} with daily monitoring and daily or hourly rebalancing; iv) the
mean value process $V$ obtained from (\ref{eq: V_i}) and (\ref{eq: V_n})
using multinomial approximation of L\'{e}vy transition probabilities; and v)
the standard deviation of the unconditional expected squared hedging error $%
\varepsilon _{0}$ obtained from (\ref{eq: eps_0}) using multinomial L\'{e}vy
transition probabilities. The barrier of an up-and-out call has to be above
the stock price for the option to be still alive, we therefore parametrize
the delta of the barrier by values starting at\footnote{%
The barrier with delta of $10^{-100}$ is so high that the corresponding
results are, for all intents and purposes, indistinguishable from a plain
vanilla option.} $10^{-100}$ and going up to $0.49$. The deltas of the
striking price range between $0.01$ and $0.99$. Numerical results for
different values of $\Delta $, $T$ and $\mu $ are shown in Tables \ref%
{Table3}-\ref{Table6}.

\begin{table}[htbp] \centering%
\begin{tabular}{rrcccccc}
\hline
&  & \multicolumn{6}{c}{\textbf{barrier (delta/level)}} \\ \cline{3-8}
&  & 1E-100 & 0.01 & 0.10 & 0.30 & 0.45 & 0.49 \\ 
&  & 343.8 & 114.6 & 107.9 & 103.3 & 100.9 & 100.3 \\ 
\multicolumn{2}{c}{strike} & \multicolumn{1}{r}{} & \multicolumn{1}{r}{} & 
\multicolumn{1}{r}{} & \multicolumn{1}{r}{} & \multicolumn{1}{r}{} & 
\multicolumn{1}{r}{} \\ 
\multicolumn{1}{c}{delta} & \multicolumn{1}{c}{level} & \multicolumn{1}{r}{}
& \multicolumn{1}{r}{} & \multicolumn{1}{r}{} & \multicolumn{1}{r}{} & 
\multicolumn{1}{r}{} & \multicolumn{1}{r}{\bigstrut[b]} \\ \cline{1-2}
\multicolumn{1}{c}{\multirow{5}[2]{*}{0.01}} & \multicolumn{1}{c}{%
\multirow{5}[2]{*}{114.6}} & 0.019 & \multicolumn{1}{r}{} & 
\multicolumn{1}{r}{} & \multicolumn{1}{r}{} & \multicolumn{1}{r}{} & 
\multicolumn{1}{r}{\bigstrut[t]} \\ 
\multicolumn{1}{c}{} & \multicolumn{1}{c}{} & 0.019 & \multicolumn{1}{r}{} & 
\multicolumn{1}{r}{} & \multicolumn{1}{r}{} & \multicolumn{1}{r}{} & 
\multicolumn{1}{r}{} \\ 
\multicolumn{1}{c}{} & \multicolumn{1}{c}{} & 0.104 & \multicolumn{1}{r}{} & 
\multicolumn{1}{r}{} & \multicolumn{1}{r}{} & \multicolumn{1}{r}{} & 
\multicolumn{1}{r}{} \\ 
\multicolumn{1}{c}{} & \multicolumn{1}{c}{} & \textbf{0.020} & 
\multicolumn{1}{r}{} & \multicolumn{1}{r}{} & \multicolumn{1}{r}{} & 
\multicolumn{1}{r}{} & \multicolumn{1}{r}{} \\ 
\multicolumn{1}{c}{} & \multicolumn{1}{c}{} & \textbf{0.122} & 
\multicolumn{1}{r}{} & \multicolumn{1}{r}{} & \multicolumn{1}{r}{} & 
\multicolumn{1}{r}{} & \multicolumn{1}{r}{\bigstrut[b]} \\ \cline{1-4}
\multicolumn{1}{c}{\multirow{5}[2]{*}{0.1}} & \multicolumn{1}{c}{%
\multirow{5}[2]{*}{107.9}} & 0.268 & 0.151 &  &  &  &  \\ 
\multicolumn{1}{c}{} & \multicolumn{1}{c}{} & 0.268 & 0.172 &  &  &  &  \\ 
\multicolumn{1}{c}{} & \multicolumn{1}{c}{} & 0.286 & 0.417 &  &  &  &  \\ 
\multicolumn{1}{c}{} & \multicolumn{1}{c}{} & \textbf{0.267} & \textbf{0.170}
&  &  &  &  \\ 
\multicolumn{1}{c}{} & \multicolumn{1}{c}{} & \textbf{0.326} & \textbf{0.443}
&  &  &  &  \\ \cline{1-5}
\multicolumn{1}{c}{\multirow{5}[2]{*}{0.3}} & \multicolumn{1}{c}{%
\multirow{5}[2]{*}{103.3}} & 1.071 & 0.874 & 0.182 &  &  &  \\ 
\multicolumn{1}{c}{} & \multicolumn{1}{c}{} & 1.071 & 0.916 & 0.255 &  &  & 
\\ 
\multicolumn{1}{c}{} & \multicolumn{1}{c}{} & 0.408 & 0.734 & 0.521 &  &  & 
\\ 
\multicolumn{1}{c}{} & \multicolumn{1}{c}{} & \textbf{1.066} & \textbf{0.910}
& \textbf{0.257} &  &  &  \\ 
\multicolumn{1}{c}{} & \multicolumn{1}{c}{} & \textbf{0.469} & \textbf{0.809}
& \textbf{0.545} &  &  &  \\ \cline{1-6}
\multicolumn{1}{c}{\multirow{5}[2]{*}{0.45}} & \multicolumn{1}{c}{%
\multirow{5}[2]{*}{100.9}} & 1.900 & 1.663 & 0.608 & 0.023 &  &  \\ 
\multicolumn{1}{c}{} & \multicolumn{1}{c}{} & 1.900 & 1.716 & 0.752 & 0.052
&  &  \\ 
\multicolumn{1}{c}{} & \multicolumn{1}{c}{} & 0.430 & 0.889 & 0.839 & 0.212
&  &  \\ 
\multicolumn{1}{c}{} & \multicolumn{1}{c}{} & \textbf{1.894} & \textbf{1.709}
& \textbf{0.759} & \textbf{0.053} &  &  \\ 
\multicolumn{1}{c}{} & \multicolumn{1}{c}{} & \textbf{0.491} & \textbf{0.960}
& \textbf{0.912} & \textbf{0.223} &  &  \\ \cline{1-7}
\multicolumn{1}{c}{\multirow{5}[2]{*}{0.49}} & \multicolumn{1}{c}{%
\multirow{5}[2]{*}{100.3}} & 2.162 & 1.915 & 0.767 & 0.044 & 0.000 &  \\ 
\multicolumn{1}{c}{} & \multicolumn{1}{c}{} & 2.162 & 1.971 & 0.930 & 0.089
& 0.001 &  \\ 
\multicolumn{1}{c}{} & \multicolumn{1}{c}{} & 0.427 & 0.913 & 0.931 & 0.290
& 0.020 &  \\ 
\multicolumn{1}{c}{} & \multicolumn{1}{c}{} & \textbf{2.156} & \textbf{1.964}
& \textbf{0.938} & \textbf{0.092} & \textbf{0.001} &  \\ 
\multicolumn{1}{c}{} & \multicolumn{1}{c}{} & \textbf{0.491} & \textbf{1.007}
& \textbf{0.991} & \textbf{0.300} & \textbf{0.020} &  \\ \hline
\multicolumn{1}{c}{\multirow{5}[2]{*}{0.75}} & \multicolumn{1}{c}{%
\multirow{5}[2]{*}{96.3}} & 4.563 & 4.247 & 2.447 & 0.515 & 0.054 & 0.013 \\ 
\multicolumn{1}{c}{} & \multicolumn{1}{c}{} & 4.563 & 4.321 & 2.750 & 0.754
& 0.141 & 0.071 \\ 
\multicolumn{1}{c}{} & \multicolumn{1}{c}{} & 0.348 & 1.135 & 1.471 & 0.836
& 0.357 & 0.252 \\ 
\multicolumn{1}{c}{} & \multicolumn{1}{c}{} & \textbf{4.560} & \textbf{4.317}
& \textbf{2.770} & \textbf{0.773} & \textbf{0.146} & \textbf{0.073} \\ 
\multicolumn{1}{c}{} & \multicolumn{1}{c}{} & \textbf{0.397} & \textbf{1.222}
& \textbf{1.610} & \textbf{0.915} & \textbf{0.379} & \textbf{0.265} \\ \hline
\multicolumn{1}{c}{\multirow{5}[2]{*}{0.99}} & \multicolumn{1}{c}{%
\multirow{5}[2]{*}{87.5}} & 12.488 & 12.020 & 8.764 & 3.512 & 0.808 & 0.262
\\ 
\multicolumn{1}{c}{} & \multicolumn{1}{c}{} & 12.488 & 12.134 & 9.385 & 4.427
& 1.598 & 1.037 \\ 
\multicolumn{1}{c}{} & \multicolumn{1}{c}{} & 0.065 & 1.607 & 2.671 & 2.138
& 1.318 & 1.092 \\ 
\multicolumn{1}{c}{} & \multicolumn{1}{c}{} & \textbf{12.489} & \textbf{%
12.134} & \textbf{9.429} & \textbf{4.491} & \textbf{1.632} & \textbf{1.050}
\\ 
\multicolumn{1}{c}{} & \multicolumn{1}{c}{} & \textbf{0.077} & \textbf{1.761}
& \textbf{2.874} & \textbf{2.316} & \textbf{1.452} & \textbf{1.193} \\ \hline
\end{tabular}%
\caption{Mean value and hedging error for a  daily monitored up-and-out call option. $T=$ 1 month, $\Delta =$ 1 day, 
$\mu=0.1,  r=0$. 
For each strike and barrier level we report 5 values: i) Black--Scholes value of continuously monitored option, 
ii) mean value for normally distributed log returns and discretely (daily) monitored option, 
iii) hedging error corresponding to ii); 
iv)  the mean value process $V_0$ for the empirical distribution of log returns(discrete monitoring); 
v) standard deviation of the unconditional hedging error corresponding to iv). 
Strike and barrier levels are parametrized by the Black--Scholes delta of their position.}%
\label{Table3}%
\end{table}%

\begin{table}[htbp] \centering%
\begin{tabular}{rrrrrrrr}
\hline
&  & \multicolumn{6}{c}{\textbf{barrier (delta/level)}} \\ \cline{3-8}
&  & \multicolumn{1}{c}{1E-100} & \multicolumn{1}{c}{0.01} & 
\multicolumn{1}{c}{0.10} & \multicolumn{1}{c}{0.30} & \multicolumn{1}{c}{0.45
} & \multicolumn{1}{c}{0.49} \\ 
&  & \multicolumn{1}{c}{2071.0} & \multicolumn{1}{c}{140.5} & 
\multicolumn{1}{c}{121.2} & \multicolumn{1}{c}{108.8} & \multicolumn{1}{c}{
102.8} & \multicolumn{1}{c}{101.4} \\ 
\multicolumn{2}{c}{strike} &  &  &  &  &  &  \\ 
\multicolumn{1}{c}{delta} & \multicolumn{1}{c}{level} &  &  &  &  &  & %
\bigstrut[b] \\ \cline{1-2}
\multicolumn{1}{c}{\multirow{5}[2]{*}{0.01}} & \multicolumn{1}{c}{%
\multirow{5}[2]{*}{140.5}} & \multicolumn{1}{c}{0.046} &  &  &  &  & %
\bigstrut[t] \\ 
\multicolumn{1}{c}{} & \multicolumn{1}{c}{} & \multicolumn{1}{c}{0.046} &  & 
&  &  &  \\ 
\multicolumn{1}{c}{} & \multicolumn{1}{c}{} & \multicolumn{1}{c}{0.152} &  & 
&  &  &  \\ 
\multicolumn{1}{c}{} & \multicolumn{1}{c}{} & \multicolumn{1}{c}{\textbf{%
0.046}} &  &  &  &  &  \\ 
\multicolumn{1}{c}{} & \multicolumn{1}{c}{} & \multicolumn{1}{c}{\textbf{%
0.176}} &  &  &  &  & \bigstrut[b] \\ \cline{1-4}
\multicolumn{1}{c}{\multirow{5}[2]{*}{0.1}} & \multicolumn{1}{c}{%
\multirow{5}[2]{*}{121.2}} & \multicolumn{1}{c}{0.635} & \multicolumn{1}{c}{
0.364} & \multicolumn{1}{c}{} & \multicolumn{1}{c}{} & \multicolumn{1}{c}{}
& \multicolumn{1}{c}{} \\ 
\multicolumn{1}{c}{} & \multicolumn{1}{c}{} & \multicolumn{1}{c}{0.635} & 
\multicolumn{1}{c}{0.387} & \multicolumn{1}{c}{} & \multicolumn{1}{c}{} & 
\multicolumn{1}{c}{} & \multicolumn{1}{c}{} \\ 
\multicolumn{1}{c}{} & \multicolumn{1}{c}{} & \multicolumn{1}{c}{0.329} & 
\multicolumn{1}{c}{0.674} & \multicolumn{1}{c}{} & \multicolumn{1}{c}{} & 
\multicolumn{1}{c}{} & \multicolumn{1}{c}{} \\ 
\multicolumn{1}{c}{} & \multicolumn{1}{c}{} & \multicolumn{1}{c}{\textbf{%
0.631}} & \multicolumn{1}{c}{\textbf{0.386}} & \multicolumn{1}{c}{} & 
\multicolumn{1}{c}{} & \multicolumn{1}{c}{} & \multicolumn{1}{c}{} \\ 
\multicolumn{1}{c}{} & \multicolumn{1}{c}{} & \multicolumn{1}{c}{\textbf{%
0.381}} & \multicolumn{1}{c}{\textbf{0.740}} & \multicolumn{1}{c}{} & 
\multicolumn{1}{c}{} & \multicolumn{1}{c}{} & \multicolumn{1}{c}{} \\ 
\cline{1-5}
\multicolumn{1}{c}{\multirow{5}[2]{*}{0.3}} & \multicolumn{1}{c}{%
\multirow{5}[2]{*}{108.8}} & \multicolumn{1}{c}{2.514} & \multicolumn{1}{c}{
2.073} & \multicolumn{1}{c}{0.447} & \multicolumn{1}{c}{} & 
\multicolumn{1}{c}{} & \multicolumn{1}{c}{} \\ 
\multicolumn{1}{c}{} & \multicolumn{1}{c}{} & \multicolumn{1}{c}{2.514} & 
\multicolumn{1}{c}{2.118} & \multicolumn{1}{c}{0.522} & \multicolumn{1}{c}{}
& \multicolumn{1}{c}{} & \multicolumn{1}{c}{} \\ 
\multicolumn{1}{c}{} & \multicolumn{1}{c}{} & \multicolumn{1}{c}{0.437} & 
\multicolumn{1}{c}{1.238} & \multicolumn{1}{c}{0.711} & \multicolumn{1}{c}{}
& \multicolumn{1}{c}{} & \multicolumn{1}{c}{} \\ 
\multicolumn{1}{c}{} & \multicolumn{1}{c}{} & \multicolumn{1}{c}{\textbf{%
2.506}} & \multicolumn{1}{c}{\textbf{2.115}} & \multicolumn{1}{c}{\textbf{%
0.526}} & \multicolumn{1}{c}{} & \multicolumn{1}{c}{} & \multicolumn{1}{c}{}
\\ 
\multicolumn{1}{c}{} & \multicolumn{1}{c}{} & \multicolumn{1}{c}{\textbf{%
0.506}} & \multicolumn{1}{c}{\textbf{1.364}} & \multicolumn{1}{c}{\textbf{%
0.778}} & \multicolumn{1}{c}{} & \multicolumn{1}{c}{} & \multicolumn{1}{c}{}
\\ \cline{1-6}
\multicolumn{1}{c}{\multirow{5}[2]{*}{0.45}} & \multicolumn{1}{c}{%
\multirow{5}[2]{*}{102.8}} & \multicolumn{1}{c}{4.434} & \multicolumn{1}{c}{
3.910} & \multicolumn{1}{c}{1.475} & \multicolumn{1}{c}{0.059} & 
\multicolumn{1}{c}{} & \multicolumn{1}{c}{} \\ 
\multicolumn{1}{c}{} & \multicolumn{1}{c}{} & \multicolumn{1}{c}{4.435} & 
\multicolumn{1}{c}{3.966} & \multicolumn{1}{c}{1.622} & \multicolumn{1}{c}{
0.087} & \multicolumn{1}{c}{} & \multicolumn{1}{c}{} \\ 
\multicolumn{1}{c}{} & \multicolumn{1}{c}{} & \multicolumn{1}{c}{0.425} & 
\multicolumn{1}{c}{1.366} & \multicolumn{1}{c}{1.214} & \multicolumn{1}{c}{
0.293} & \multicolumn{1}{c}{} & \multicolumn{1}{c}{} \\ 
\multicolumn{1}{c}{} & \multicolumn{1}{c}{} & \multicolumn{1}{c}{\textbf{%
4.426}} & \multicolumn{1}{c}{\textbf{3.963}} & \multicolumn{1}{c}{\textbf{%
1.632}} & \multicolumn{1}{c}{\textbf{0.088}} & \multicolumn{1}{c}{} & 
\multicolumn{1}{c}{} \\ 
\multicolumn{1}{c}{} & \multicolumn{1}{c}{} & \multicolumn{1}{c}{\textbf{%
0.493}} & \multicolumn{1}{c}{\textbf{1.505}} & \multicolumn{1}{c}{\textbf{%
1.335}} & \multicolumn{1}{c}{\textbf{0.316}} & \multicolumn{1}{c}{} & 
\multicolumn{1}{c}{} \\ \cline{1-7}
\multicolumn{1}{c}{\multirow{5}[2]{*}{0.49}} & \multicolumn{1}{c}{%
\multirow{5}[2]{*}{101.4}} & \multicolumn{1}{c}{5.038} & \multicolumn{1}{c}{
4.493} & \multicolumn{1}{c}{1.854} & \multicolumn{1}{c}{0.114} & 
\multicolumn{1}{c}{0.000} & \multicolumn{1}{c}{} \\ 
\multicolumn{1}{c}{} & \multicolumn{1}{c}{} & \multicolumn{1}{c}{5.039} & 
\multicolumn{1}{c}{4.552} & \multicolumn{1}{c}{2.020} & \multicolumn{1}{c}{
0.157} & \multicolumn{1}{c}{0.001} & \multicolumn{1}{c}{} \\ 
\multicolumn{1}{c}{} & \multicolumn{1}{c}{} & \multicolumn{1}{c}{0.428} & 
\multicolumn{1}{c}{1.526} & \multicolumn{1}{c}{1.226} & \multicolumn{1}{c}{
0.350} & \multicolumn{1}{c}{0.029} & \multicolumn{1}{c}{} \\ 
\multicolumn{1}{c}{} & \multicolumn{1}{c}{} & \multicolumn{1}{c}{\textbf{%
5.030}} & \multicolumn{1}{c}{\textbf{4.549}} & \multicolumn{1}{c}{\textbf{%
2.032}} & \multicolumn{1}{c}{\textbf{0.160}} & \multicolumn{1}{c}{\textbf{%
0.001}} & \multicolumn{1}{c}{} \\ 
\multicolumn{1}{c}{} & \multicolumn{1}{c}{} & \multicolumn{1}{c}{\textbf{%
0.496}} & \multicolumn{1}{c}{\textbf{1.679}} & \multicolumn{1}{c}{\textbf{%
1.351}} & \multicolumn{1}{c}{\textbf{0.380}} & \multicolumn{1}{c}{\textbf{%
0.030}} & \multicolumn{1}{c}{} \\ \hline
\multicolumn{1}{c}{\multirow{5}[2]{*}{0.75}} & \multicolumn{1}{c}{%
\multirow{5}[2]{*}{91.8}} & \multicolumn{1}{c}{10.490} & \multicolumn{1}{c}{
9.812} & \multicolumn{1}{c}{5.793} & \multicolumn{1}{c}{1.296} & 
\multicolumn{1}{c}{0.163} & \multicolumn{1}{c}{0.055} \\ 
\multicolumn{1}{c}{} & \multicolumn{1}{c}{} & \multicolumn{1}{c}{10.491} & 
\multicolumn{1}{c}{9.888} & \multicolumn{1}{c}{6.094} & \multicolumn{1}{c}{
1.525} & \multicolumn{1}{c}{0.240} & \multicolumn{1}{c}{0.102} \\ 
\multicolumn{1}{c}{} & \multicolumn{1}{c}{} & \multicolumn{1}{c}{0.314} & 
\multicolumn{1}{c}{1.746} & \multicolumn{1}{c}{2.040} & \multicolumn{1}{c}{
1.013} & \multicolumn{1}{c}{0.361} & \multicolumn{1}{c}{0.234} \\ 
\multicolumn{1}{c}{} & \multicolumn{1}{c}{} & \multicolumn{1}{c}{\textbf{%
10.485}} & \multicolumn{1}{c}{\textbf{9.890}} & \multicolumn{1}{c}{\textbf{%
6.120}} & \multicolumn{1}{c}{\textbf{1.543}} & \multicolumn{1}{c}{\textbf{%
0.246}} & \multicolumn{1}{c}{\textbf{0.105}} \\ 
\multicolumn{1}{c}{} & \multicolumn{1}{c}{} & \multicolumn{1}{c}{\textbf{%
0.364}} & \multicolumn{1}{c}{\textbf{1.921}} & \multicolumn{1}{c}{\textbf{%
2.251}} & \multicolumn{1}{c}{\textbf{1.118}} & \multicolumn{1}{c}{\textbf{%
0.398}} & \multicolumn{1}{c}{\textbf{0.258}} \\ \hline
\multicolumn{1}{c}{\multirow{5}[2]{*}{0.99}} & \multicolumn{1}{c}{%
\multirow{5}[2]{*}{72.6}} & \multicolumn{1}{c}{27.452} & \multicolumn{1}{c}{
26.507} & \multicolumn{1}{c}{19.689} & \multicolumn{1}{c}{8.316} & 
\multicolumn{1}{c}{2.280} & \multicolumn{1}{c}{1.036} \\ 
\multicolumn{1}{c}{} & \multicolumn{1}{c}{} & \multicolumn{1}{c}{27.452} & 
\multicolumn{1}{c}{26.618} & \multicolumn{1}{c}{20.272} & \multicolumn{1}{c}{
9.155} & \multicolumn{1}{c}{2.972} & \multicolumn{1}{c}{1.659} \\ 
\multicolumn{1}{c}{} & \multicolumn{1}{c}{} & \multicolumn{1}{c}{0.048} & 
\multicolumn{1}{c}{2.605} & \multicolumn{1}{c}{3.244} & \multicolumn{1}{c}{
2.244} & \multicolumn{1}{c}{1.307} & \multicolumn{1}{c}{0.989} \\ 
\multicolumn{1}{c}{} & \multicolumn{1}{c}{} & \multicolumn{1}{c}{\textbf{%
27.452}} & \multicolumn{1}{c}{\textbf{26.628}} & \multicolumn{1}{c}{\textbf{%
20.326}} & \multicolumn{1}{c}{\textbf{9.217}} & \multicolumn{1}{c}{\textbf{%
3.012}} & \multicolumn{1}{c}{\textbf{1.689}} \\ 
\multicolumn{1}{c}{} & \multicolumn{1}{c}{} & \multicolumn{1}{c}{\textbf{%
0.056}} & \multicolumn{1}{c}{\textbf{2.858}} & \multicolumn{1}{c}{\textbf{%
3.589}} & \multicolumn{1}{c}{\textbf{2.494}} & \multicolumn{1}{c}{\textbf{%
1.455}} & \multicolumn{1}{c}{\textbf{1.104}} \\ \hline
\end{tabular}%
\caption{Mean value and hedging error for a  daily monitored up-and-out call option. $T=$ 6 month, $\Delta =$ 1 day, 
$\mu=0.1,  r=0$. 
For each strike and barrier level we report 5 values: i) Black--Scholes value of continuously monitored option, 
ii) mean value for normally distributed log returns and discretely (daily) monitored option, 
iii) hedging error corresponding to ii); 
iv)  the mean value process $V_0$ for the empirical distribution of log returns(discrete monitoring); 
v) standard deviation of the unconditional hedging error corresponding to iv). 
Strike and barrier levels are parametrized by the Black--Scholes delta of their position.}%
\label{Table4}%
\end{table}%

\begin{table}[htbp] \centering%
\begin{tabular}{rrcccccc}
\hline
&  & \multicolumn{6}{c}{\textbf{barrier (delta/level)}} \\ \cline{3-8}
&  & 1E-100 & 0.01 & 0.10 & 0.30 & 0.45 & 0.49 \\ 
&  & 2070.99 & 140.55 & 121.17 & 108.82 & 102.83 & 101.37 \\ 
\multicolumn{2}{c}{strike} & \multicolumn{1}{r}{} & \multicolumn{1}{r}{} & 
\multicolumn{1}{r}{} & \multicolumn{1}{r}{} & \multicolumn{1}{r}{} & 
\multicolumn{1}{r}{} \\ 
\multicolumn{1}{c}{delta} & \multicolumn{1}{c}{level} & \multicolumn{1}{r}{}
& \multicolumn{1}{r}{} & \multicolumn{1}{r}{} & \multicolumn{1}{r}{} & 
\multicolumn{1}{r}{} & \multicolumn{1}{r}{\bigstrut[b]} \\ \cline{1-2}
\multicolumn{1}{c}{\multirow{5}[2]{*}{0.01}} & \multicolumn{1}{c}{%
\multirow{5}[2]{*}{140.55}} & 0.046 & \multicolumn{1}{r}{} & 
\multicolumn{1}{r}{} & \multicolumn{1}{r}{} & \multicolumn{1}{r}{} & 
\multicolumn{1}{r}{\bigstrut[t]} \\ 
\multicolumn{1}{c}{} & \multicolumn{1}{c}{} & 0.046 & \multicolumn{1}{r}{} & 
\multicolumn{1}{r}{} & \multicolumn{1}{r}{} & \multicolumn{1}{r}{} & 
\multicolumn{1}{r}{} \\ 
\multicolumn{1}{c}{} & \multicolumn{1}{c}{} & 0.071 & \multicolumn{1}{r}{} & 
\multicolumn{1}{r}{} & \multicolumn{1}{r}{} & \multicolumn{1}{r}{} & 
\multicolumn{1}{r}{} \\ 
\multicolumn{1}{c}{} & \multicolumn{1}{c}{} & \textbf{0.047} & 
\multicolumn{1}{r}{} & \multicolumn{1}{r}{} & \multicolumn{1}{r}{} & 
\multicolumn{1}{r}{} & \multicolumn{1}{r}{} \\ 
\multicolumn{1}{c}{} & \multicolumn{1}{c}{} & \textbf{0.084} & 
\multicolumn{1}{r}{} & \multicolumn{1}{r}{} & \multicolumn{1}{r}{} & 
\multicolumn{1}{r}{} & \multicolumn{1}{r}{\bigstrut[b]} \\ \cline{1-4}
\multicolumn{1}{c}{\multirow{5}[2]{*}{0.1}} & \multicolumn{1}{c}{%
\multirow{5}[2]{*}{121.17}} & 0.635 & 0.364 &  &  &  &  \\ 
\multicolumn{1}{c}{} & \multicolumn{1}{c}{} & 0.635 & 0.387 &  &  &  &  \\ 
\multicolumn{1}{c}{} & \multicolumn{1}{c}{} & 0.236 & 0.351 &  &  &  &  \\ 
\multicolumn{1}{c}{} & \multicolumn{1}{c}{} & \textbf{0.637} & \textbf{0.387}
&  &  &  &  \\ 
\multicolumn{1}{c}{} & \multicolumn{1}{c}{} & \textbf{0.270} & \textbf{0.377}
&  &  &  &  \\ \cline{1-5}
\multicolumn{1}{c}{\multirow{5}[2]{*}{0.3}} & \multicolumn{1}{c}{%
\multirow{5}[2]{*}{108.82}} & 2.514 & 2.073 & 0.447 &  &  &  \\ 
\multicolumn{1}{c}{} & \multicolumn{1}{c}{} & 2.514 & 2.118 & 0.522 &  &  & 
\\ 
\multicolumn{1}{c}{} & \multicolumn{1}{c}{} & 0.374 & 0.612 & 0.515 &  &  & 
\\ 
\multicolumn{1}{c}{} & \multicolumn{1}{c}{} & \textbf{2.514} & \textbf{2.116}
& \textbf{0.524} &  &  &  \\ 
\multicolumn{1}{c}{} & \multicolumn{1}{c}{} & \textbf{0.438} & \textbf{0.705}
& \textbf{0.546} &  &  &  \\ \cline{1-6}
\multicolumn{1}{c}{\multirow{5}[2]{*}{0.45}} & \multicolumn{1}{c}{%
\multirow{5}[2]{*}{102.83}} & 4.434 & 3.910 & 1.475 & 0.059 &  &  \\ 
\multicolumn{1}{c}{} & \multicolumn{1}{c}{} & 4.435 & 3.966 & 1.622 & 0.087
&  &  \\ 
\multicolumn{1}{c}{} & \multicolumn{1}{c}{} & 0.419 & 0.755 & 0.798 & 0.239
&  &  \\ 
\multicolumn{1}{c}{} & \multicolumn{1}{c}{} & \textbf{4.434} & \textbf{3.962}
& \textbf{1.626} & \textbf{0.088} &  &  \\ 
\multicolumn{1}{c}{} & \multicolumn{1}{c}{} & \textbf{0.481} & \textbf{0.822}
& \textbf{0.913} & \textbf{0.268} &  &  \\ \cline{1-7}
\multicolumn{1}{c}{\multirow{5}[1]{*}{0.49}} & \multicolumn{1}{c}{%
\multirow{5}[1]{*}{101.37}} & 5.038 & 4.493 & 1.854 & 0.114 & 0.000 &  \\ 
\multicolumn{1}{c}{} & \multicolumn{1}{c}{} & 5.039 & 4.552 & 2.020 & 0.157
& 0.001 &  \\ 
\multicolumn{1}{c}{} & \multicolumn{1}{c}{} & 0.419 & 0.761 & 0.901 & 0.326
& 0.026 &  \\ 
\multicolumn{1}{c}{} & \multicolumn{1}{c}{} & \textbf{5.037} & \textbf{4.548}
& \textbf{2.024} & \textbf{0.159} & \textbf{0.001} &  \\ 
\multicolumn{1}{c}{} & \multicolumn{1}{c}{} & \textbf{0.490} & \textbf{0.876}
& \textbf{0.968} & \textbf{0.342} & \textbf{0.028} &  \\ \hline
\multicolumn{1}{c}{\multirow{5}[0]{*}{0.75}} & \multicolumn{1}{c}{%
\multirow{5}[0]{*}{91.79}} & 10.490 & 9.812 & 5.793 & 1.296 & 0.163 & 0.055
\\ 
\multicolumn{1}{c}{} & \multicolumn{1}{c}{} & 10.491 & 9.888 & 6.094 & 1.525
& 0.240 & 0.102 \\ 
\multicolumn{1}{c}{} & \multicolumn{1}{c}{} & 0.380 & 0.939 & 1.367 & 0.866
& 0.391 & 0.262 \\ 
\multicolumn{1}{c}{} & \multicolumn{1}{c}{} & \textbf{10.489} & \textbf{9.883%
} & \textbf{6.104} & \textbf{1.538} & \textbf{0.245} & \textbf{0.105} \\ 
\multicolumn{1}{c}{} & \multicolumn{1}{c}{} & \textbf{0.435} & \textbf{1.019}
& \textbf{1.564} & \textbf{0.992} & \textbf{0.419} & \textbf{0.281} \\ \hline
\multicolumn{1}{c}{\multirow{5}[1]{*}{0.99}} & \multicolumn{1}{c}{%
\multirow{5}[1]{*}{72.60}} & 27.452 & 26.507 & 19.689 & 8.316 & 2.280 & 1.036
\\ 
\multicolumn{1}{c}{} & \multicolumn{1}{c}{} & 27.452 & 26.618 & 20.272 & 
9.155 & 2.972 & 1.659 \\ 
\multicolumn{1}{c}{} & \multicolumn{1}{c}{} & 0.094 & 1.199 & 2.382 & 2.109
& 1.300 & 1.019 \\ 
\multicolumn{1}{c}{} & \multicolumn{1}{c}{} & \textbf{27.452} & \textbf{%
26.613} & \textbf{20.294} & \textbf{9.205} & \textbf{3.014} & \textbf{1.694}
\\ 
\multicolumn{1}{c}{} & \multicolumn{1}{c}{} & \textbf{0.111} & \textbf{1.371}
& \textbf{2.591} & \textbf{2.306} & \textbf{1.487} & \textbf{1.165} \\ \hline
\end{tabular}%
\caption{Mean value and hedging error for a daily monitored up-and-out call option. $T=$ 6 months, 
$\Delta =$ 1 day, $\mu=-0.1,  r=0$. 
For each strike and barrier level we report 5 values: i) Black--Scholes value of continuously monitored option, 
ii) mean value for normally distributed log returns and discretely (daily) monitored option, 
iii) hedging error corresponding to ii); 
iv)  the mean value process $V_0$ for the empirical distribution of log returns(discrete monitoring); 
v) standard deviation of the unconditional hedging error corresponding to iv). 
Strike and barrier levels are parametrized by the Black--Scholes delta of their position.}%
\label{Table5}%
\end{table}%

\begin{table}[htbp] \centering%
\begin{tabular}{rrcccccc}
\hline
&  & \multicolumn{6}{c}{\textbf{barrier (delta/level)}} \\ \cline{3-8}
&  & 1E-100 & 0.01 & 0.10 & 0.30 & 0.45 & 0.49 \\ 
&  & 342.09 & 114.57 & 107.86 & 103.25 & 100.90 & 100.31 \\ 
\multicolumn{2}{c}{strike} & \multicolumn{1}{r}{} & \multicolumn{1}{r}{} & 
\multicolumn{1}{r}{} & \multicolumn{1}{r}{} & \multicolumn{1}{r}{} & 
\multicolumn{1}{r}{} \\ 
\multicolumn{1}{c}{delta} & \multicolumn{1}{c}{level} & \multicolumn{1}{r}{}
& \multicolumn{1}{r}{} & \multicolumn{1}{r}{} & \multicolumn{1}{r}{} & 
\multicolumn{1}{r}{} & \multicolumn{1}{r}{\bigstrut[b]} \\ \cline{1-2}
\multicolumn{1}{c}{\multirow{5}[2]{*}{0.01}} & \multicolumn{1}{c}{%
\multirow{5}[2]{*}{114.63}} & 0.019 & \multicolumn{1}{r}{} & 
\multicolumn{1}{r}{} & \multicolumn{1}{r}{} & \multicolumn{1}{r}{} & 
\multicolumn{1}{r}{\bigstrut[t]} \\ 
\multicolumn{1}{c}{} & \multicolumn{1}{c}{} & 0.019 & \multicolumn{1}{r}{} & 
\multicolumn{1}{r}{} & \multicolumn{1}{r}{} & \multicolumn{1}{r}{} & 
\multicolumn{1}{r}{} \\ 
\multicolumn{1}{c}{} & \multicolumn{1}{c}{} & 0.039 & \multicolumn{1}{r}{} & 
\multicolumn{1}{r}{} & \multicolumn{1}{r}{} & \multicolumn{1}{r}{} & 
\multicolumn{1}{r}{} \\ 
\multicolumn{1}{c}{} & \multicolumn{1}{c}{} & \textbf{0.020} & 
\multicolumn{1}{r}{} & \multicolumn{1}{r}{} & \multicolumn{1}{r}{} & 
\multicolumn{1}{r}{} & \multicolumn{1}{r}{} \\ 
\multicolumn{1}{c}{} & \multicolumn{1}{c}{} & \textbf{0.074} & 
\multicolumn{1}{r}{} & \multicolumn{1}{r}{} & \multicolumn{1}{r}{} & 
\multicolumn{1}{r}{} & \multicolumn{1}{r}{\bigstrut[b]} \\ \cline{1-4}
\multicolumn{1}{c}{\multirow{5}[2]{*}{0.1}} & \multicolumn{1}{c}{%
\multirow{5}[2]{*}{107.89}} & 0.268 & 0.151 &  &  &  &  \\ 
\multicolumn{1}{c}{} & \multicolumn{1}{c}{} & 0.268 & 0.172 &  &  &  &  \\ 
\multicolumn{1}{c}{} & \multicolumn{1}{c}{} & 0.104 & 0.214 &  &  &  &  \\ 
\multicolumn{1}{c}{} & \multicolumn{1}{c}{} & \textbf{0.268} & \textbf{0.171}
&  &  &  &  \\ 
\multicolumn{1}{c}{} & \multicolumn{1}{c}{} & \textbf{0.197} & \textbf{0.303}
&  &  &  &  \\ \cline{1-5}
\multicolumn{1}{c}{\multirow{5}[2]{*}{0.3}} & \multicolumn{1}{c}{%
\multirow{5}[2]{*}{103.26}} & 1.071 & 0.874 & 0.182 &  &  &  \\ 
\multicolumn{1}{c}{} & \multicolumn{1}{c}{} & 1.071 & 0.916 & 0.255 &  &  & 
\\ 
\multicolumn{1}{c}{} & \multicolumn{1}{c}{} & 0.149 & 0.380 & 0.279 &  &  & 
\\ 
\multicolumn{1}{c}{} & \multicolumn{1}{c}{} & \textbf{1.068} & \textbf{0.912}
& \textbf{0.257} &  &  &  \\ 
\multicolumn{1}{c}{} & \multicolumn{1}{c}{} & \textbf{0.282} & \textbf{0.547}
& \textbf{0.381} &  &  &  \\ \cline{1-6}
\multicolumn{1}{c}{\multirow{5}[2]{*}{0.45}} & \multicolumn{1}{c}{%
\multirow{5}[2]{*}{100.90}} & 1.900 & 1.663 & 0.608 & 0.023 &  &  \\ 
\multicolumn{1}{c}{} & \multicolumn{1}{c}{} & 1.900 & 1.716 & 0.752 & 0.052
&  &  \\ 
\multicolumn{1}{c}{} & \multicolumn{1}{c}{} & 0.155 & 0.453 & 0.451 & 0.129
&  &  \\ 
\multicolumn{1}{c}{} & \multicolumn{1}{c}{} & \textbf{1.897} & \textbf{1.710}
& \textbf{0.758} & \textbf{0.053} &  &  \\ 
\multicolumn{1}{c}{} & \multicolumn{1}{c}{} & \textbf{0.295} & \textbf{0.651}
& \textbf{0.628} & \textbf{0.165} &  &  \\ \cline{1-7}
\multicolumn{1}{c}{\multirow{5}[1]{*}{0.49}} & \multicolumn{1}{c}{%
\multirow{5}[1]{*}{100.31}} & 2.162 & 1.915 & 0.767 & 0.044 & 0.000 &  \\ 
\multicolumn{1}{c}{} & \multicolumn{1}{c}{} & 2.162 & 1.971 & 0.930 & 0.089
& 0.001 &  \\ 
\multicolumn{1}{c}{} & \multicolumn{1}{c}{} & 0.155 & 0.478 & 0.487 & 0.164
& 0.017 &  \\ 
\multicolumn{1}{c}{} & \multicolumn{1}{c}{} & \textbf{2.159} & \textbf{1.965}
& \textbf{0.937} & \textbf{0.092} & \textbf{0.001} &  \\ 
\multicolumn{1}{c}{} & \multicolumn{1}{c}{} & \textbf{0.295} & \textbf{0.684}
& \textbf{0.681} & \textbf{0.216} & \textbf{0.018} &  \\ \hline
\multicolumn{1}{c}{\multirow{5}[0]{*}{0.75}} & \multicolumn{1}{c}{%
\multirow{5}[0]{*}{96.33}} & 4.563 & 4.247 & 2.447 & 0.515 & 0.054 & 0.013
\\ 
\multicolumn{1}{c}{} & \multicolumn{1}{c}{} & 4.563 & 4.321 & 2.750 & 0.754
& 0.141 & 0.071 \\ 
\multicolumn{1}{c}{} & \multicolumn{1}{c}{} & 0.126 & 0.595 & 0.784 & 0.451
& 0.193 & 0.137 \\ 
\multicolumn{1}{c}{} & \multicolumn{1}{c}{} & \textbf{4.562} & \textbf{4.317}
& \textbf{2.767} & \textbf{0.771} & \textbf{0.145} & \textbf{0.073} \\ 
\multicolumn{1}{c}{} & \multicolumn{1}{c}{} & \textbf{0.240} & \textbf{0.839}
& \textbf{1.103} & \textbf{0.631} & \textbf{0.265} & \textbf{0.186} \\ \hline
\multicolumn{1}{c}{\multirow{5}[1]{*}{0.99}} & \multicolumn{1}{c}{%
\multirow{5}[1]{*}{87.53}} & 12.488 & 12.020 & 8.764 & 3.512 & 0.808 & 0.262
\\ 
\multicolumn{1}{c}{} & \multicolumn{1}{c}{} & 12.488 & 12.134 & 9.385 & 4.427
& 1.598 & 1.037 \\ 
\multicolumn{1}{c}{} & \multicolumn{1}{c}{} & 0.024 & 0.883 & 1.409 & 1.144
& 0.734 & 0.606 \\ 
\multicolumn{1}{c}{} & \multicolumn{1}{c}{} & \textbf{12.489} & \textbf{%
12.132} & \textbf{9.420} & \textbf{4.484} & \textbf{1.630} & \textbf{1.049}
\\ 
\multicolumn{1}{c}{} & \multicolumn{1}{c}{} & \textbf{0.047} & \textbf{1.224}
& \textbf{1.975} & \textbf{1.604} & \textbf{1.020} & \textbf{0.837} \\ \hline
\end{tabular}
\caption{Mean value and hedging error for a  daily monitored up-and-out call option. $T=$ 1 month, 
$\Delta =$ 1 hour, $\mu=0.1,  r=0$. 
For each strike and barrier level we report 5 values: i) Black--Scholes value of continuously monitored option, 
ii) mean value for normally distributed log returns and discretely (daily) monitored option, 
iii) hedging error corresponding to ii); 
iv)  the mean value process $V_0$ for the empirical distribution of log returns(discrete monitoring); 
v) standard deviation of the unconditional hedging error corresponding to iv). 
Strike and barrier levels are parametrized by the Black--Scholes delta of their position.}%
\label{Table6}%
\end{table}

We commence with the base case parameters $\Delta =1$ day$,$ $\mu =0.1$ in
Tables \ref{Table3} and \ref{Table4}. The mean value process $V$ coincides
to a large extent with the Black-Scholes value of a discretely monitored
option. This is a striking result, since the model in which $V$ is computed
is substantially incomplete, whereas the reasoning behind $\hat{V}$ relies
on continuous rebalancing and perfect replication. For $T=1$ month (Table %
\ref{Table3}) the difference between $V$ and $\hat{V}$ is always less than
6.4 cents in absolute value, and in relative terms it is less than 3.6\%
across all strikes and barrier levels.

The difference between $V$ and $\hat{V}$ tends to diminish with increasing
maturity. For $T=6$ months (Table \ref{Table4}) the difference between $V$
and $\hat{V}$ is less than 6.1 cents in absolute value, and less than 2.7\%
in relative terms. The signs of $V-\hat{V}$ follow a pattern across strikes
and barrier levels whereby the difference tends to be negative for very high
barrier levels in combination with high strike prices, and to be positive
elsewhere.

Let us now turn to the hedging errors. Hedging errors of barrier options
(columns 4-8) behave differently to those of plain vanilla options (column
3). The hedging error of plain vanilla options are the largest at the money
and become smaller for deep-in and deep-out-of-the-money options. In
contrast, the hedging error of an up-and-out barrier option increases with
decreasing strike price. This happens because for vanilla options the only
source of the hedging error is the non-linearity of option pay-off around
the strike price, whereas for barriers the main source of the hedging errors
is the barrier itself. The lower the strike the higher the pay-off near the
barrier and the higher the hedging errors.

Consider an (at-the-money) plain vanilla option with $T=1$ month to maturity
and strike at 100.3 (see Table \ref{Table3}, column 3). The Black-Scholes
value of this option is 2.162, and the standard deviation of the
unconditional hedging error is 0.427, due to daily rebalancing. If we
consider the empirical distribution of log returns, which exhibits excess
kurtosis, the hedging error increases to 0.491. Take now a barrier option
with the same strike, and barrier at 107.9. The Black--Scholes price of the
barrier option is less than a half at 0.930 but the standard deviation of
the hedging error is more than double at 0.931. Thus if selling a plain
vanilla option at the Black--Scholes price based on historical volatility is
not a profitable enterprise, doing the same for barrier options is
positively counterproductive. This conclusion is more pronounced for longer
maturities and lower strikes, see Table \ref{Table4} ($T=6$ months).

Next we examine the effect of the change in the market direction, by
contrasting Table \ref{Table4} $(\mu =0.1)$ with Table \ref{Table5} ($\mu
=-0.1)$. The difference between the Black--Scholes no-arbitrage price of a
daily monitored barrier option $\hat{V}$ and the mean value process $V$
remains small. The mean value is higher in the bear market for plain vanilla
options (column 3) but it is generally marginally lower for barrier options,
with the exception of very low strikes in combination with very low barrier
levels. The difference in absolute value is less than two cents and less
than 1\% in relative terms (with the exception of the two vanilla option
with highest strikes). We conclude that $V\ $is largely insensitive to the
changes in $\mu $ and that the Black--Scholes price $\hat{V}$ is a very good
proxy for $V$.

The change in the market direction has a more dramatic effect on the size of
unconditional hedging errors. Recall that the standard deviation of the
unconditional error is given as a weighted average of one-period hedging
errors,%
\begin{eqnarray*}
\varepsilon _{0}^{2}(\varphi ) &=&\sum_{j=0}^{n-1}\left( R^{2}b\right)
^{n-j-1}\mathrm{E}\left[ \psi _{j}\right] , \\
\psi _{j} &=&\mathrm{Var}_{j}\left( V_{j+1}^{2}\right) -\frac{\left( \mathrm{%
Cov}_{j}\left( S_{j+1},V_{j+1}\right) \right) ^{2}}{\mathrm{Var}_{j}\left(
S_{j+1}\right) },
\end{eqnarray*}%
where $R$ and $b$ are close to $1.$ Since $V$ is largely insensitive to the
value of $\mu $ the values of $\psi $ (as a function of time, stock price
and option status) will very much coincide between the bull and the bear
market. What will be different is the \emph{expectation} of $\psi $.

The instantaneous hedging error $\psi $ arises from two non-linearities in
the option pay-off -- one around the strike price and one along the barrier.
The hedging error along the barrier tends to be more significant unless the
barrier is either very far away from the stock price or the option is just
about to be knocked out. In a bull market prices rise on average and the
barrier, being above the initial stock price, contributes more significantly
to $\mathrm{E}\left[ \psi _{j}\right] $. $\mathrm{E}\left[ \psi _{j}\right] $
will also contain more significant contribution from the strike region if
the option is initially out of the money. In contrast, in a bear market
price falls on average and $\mathrm{E}\left[ \psi _{j}\right] $ will put
less weight on the barrier region. It will contain a more significant
contribution from the strike price, if the option is in the money to begin
with. For barrier deltas equal to $10^{-100}$ and 0.49 we expect the strike
region to dominate and therefore the hedging errors in the bear market to be
larger for in-the-money options. This intuition is borne out by the
numerical results shown in Tables \ref{Table4} and \ref{Table5}.

\subsection{Asymptotics}

Let us now examine examine the effect of more frequent rebalancing by
considering $\Delta =1$ hour (Table \ref{Table6}). Although hedging now
occurs \emph{hourly} we maintain the \emph{daily monitoring} frequency of
the barrier to make the results comparable with those in Table \ref{Table3}. 

In the Black--Scholes model the \emph{standard deviation} of the hedging
error for \emph{plain vanilla} options decreases with the square root of
rebalancing interval, see \cite{BoyEma80} and \cite{Toft96}. With hourly
rebalancing this implies standard deviation equal to $\sqrt{1/8}\approx
\allowbreak 35\%$ of the daily figure (with 8-hour trading day). The
theoretical prediction turns out to be accurate, as can be seen by comparing
entries marked ii) in each row of column 3 of Tables \ref{Table3} and \ref%
{Table6} which yields the range $36\%$ to $38\%$ across all strikes.

In the empirical L\'{e}vy model the standard deviation of the unconditional
hedging error of plain vanilla options is seen to decay more slowly, see
entries marked v) in each row of column 3 of Tables \ref{Table3} and \ref%
{Table6}. With hourly rebalancing it is in the range 60\%-62\% of the daily
rebalancing figures across all strikes. In this instance the higher
frequency of hedging is (partially) offset by higher kurtosis of hourly
returns. \cite{cerny.book.09}, Section 13.7, derives an approximation of
the hedging error for leptokurtic returns and shows that rebalancing
interval must be multiplied by kurtosis minus one to obtain the correct
scaling of hedging errors. In our case Table \ref{Table1} shows the kurtosis
of daily returns is 3.72 and the kurtosis of hourly returns is 8.73, thus we
should expect hourly errors to equal $\sqrt{1/8\times 7.73/2.72}\approx
60\%\,$of the daily errors which matches the actual range of 60\% to 62\%
mentioned earlier.

Table \ref{Table1} compares the kurtosis of returns and log returns in the
calibrated L\'{e}vy model with the kurtosis achieved in its multinomial
lattice approximation. The last two columns show the number of standard
deviations of one-period log return (rounded up to the nearest quarter)
corresponding to the $10^{-5}$ and $1-10^{-5}$ quantiles of the one-period
log return distribution. This is the range represented by the lattice
approximation of the L\'{e}vy process. As an aside, we observe that the
lattice begins to struggle to approximate the kurtosis of the L\'{e}vy
process well at the 5-minute rebalancing interval.

\begin{table}[tbp] \centering%
\begin{tabular}{lllllll}
\hline
& \multicolumn{2}{l}{kurtosis lattice} & \multicolumn{2}{l}{kurtosis L\'{e}vy
} &  &  \\ 
$\Delta $ & log & level & log & level & $\frac{n_{\mathrm{down}}\eta }{%
\sigma }$ & $\frac{n_{\mathrm{up}}\eta }{\sigma }$ \\ \hline
\textbf{5 min} & \textbf{69.11} & \textbf{69.05} & \textbf{72.54} & \textbf{%
72.51} & \textbf{27} & \textbf{25} \\ 
15 min & 25.77 & 25.76 & 26.18 & 26.17 & 17 & 16 \\ 
30 min & 14.42 & 14.42 & 14.59 & 14.59 & 12.75 & 12 \\ 
\textbf{1 hr} & \textbf{8.73} & \textbf{8.73} & \textbf{8.79} & \textbf{8.80}
& \textbf{10} & \textbf{9.5} \\ 
2 hr & 5.86 & 5.87 & 5.90 & 5.90 & 8 & 7.75 \\ 
4 hr & 4.44 & 4.44 & 4.45 & 4.45 & 6.75 & 6.75 \\ 
\textbf{1 day (8 hr)} & \textbf{3.72} & \textbf{3.72} & \textbf{3.72} & 
\textbf{3.73} & \textbf{5.75} & \textbf{5.75} \\ \hline
\end{tabular}%
\caption{Kurtosis as a function of rebalancing interval}\label{Table1}%
\end{table}%

For \emph{barrier options} (columns 4-8 of Tables \ref{Table3} and \ref%
{Table6}) the Black-Scholes situation is more complicated because part of
the error is caused by the barrier itself and this part has different $%
\Delta $-asymptotics. Conjecturing that the barrier contributes an error
whose \emph{variance }is proportional to the square root of rebalancing
interval, see \cite{gobet.temam.01}, and assuming that fraction $\alpha $
of the error is generated by the strike region and the rest by the barrier,
the approximate expression for the hourly total error as a fraction of daily
error would read%
\begin{equation}
\sqrt{0.35\alpha +\sqrt{0.35}(1-\alpha )}.  \label{eq: heuristic1}
\end{equation}%
For barrier options in columns 4 and 5 of Tables \ref{Table3} and \ref%
{Table6} the percentage reduction in hedging error in the Black-Scholes
model stands between 51\% and 54\% which implies $\alpha $ values in formula
(\ref{eq: heuristic1}) between 0.25 and 0.4. Variability of $\alpha $ is to
be expected since the relative importance of the two types of errors will
depend on barrier and strike levels.

One can conjecture that for barrier options in the presence of excess
kurtosis the formula (\ref{eq: heuristic1}) will remain the same, only the
time scaling factor will be adjusted for excess kurtosis from $0.35$ to $0.6$
as in the case of plain vanilla options. We thus expect the ratio of hourly
to daily errors in the L\'{e}vy model to be 
\begin{equation}
\sqrt{0.6\alpha +\sqrt{0.6}(1-\alpha )}.  \label{eq: heuristic2}
\end{equation}%
With $\alpha $ in the range 0.25 to 0.4 heuristic (\ref{eq: heuristic2})
predicts error reduction in the range 71\%-74\% while the actual figures
from columns 4 and 5 of Tables \ref{Table3} and \ref{Table6} yield the range
68\%-70\%, which for practical purposes is a perfectly adequate
approximation.

Table \ref{Table2} provides 5-minute error data for one specific
strike/barrier combination corresponding to $\alpha =0.25.$ It reports the
hedging error $\varepsilon _{0}$ obtained from (\ref{eq: eps_0}) and the
mean value $V_{0}$ obtained from (\ref{eq: V_i}) and (\ref{eq: V_n}) using
the multinomial approximation of the empirical L\'{e}vy process and
analogous quantitities $\hat{\varepsilon}_{0}$ and $\hat{V}_{0}$ obtained
from a multinomial approximation of the Black-Scholes model.

The Black-Scholes 5-minute time scaling factor is $1/(8\times 12)=1/96=0.0104
$ and the heuristic (\ref{eq: heuristic1}) yields error reduction ratio of 
$$\sqrt{0.0104\times 0.25+\sqrt{0.0104}(1-0.25)}\approx 28\%,$$ 
while in Table \ref{Table2} we find this ratio to be $0.142/0.519\approx 27\%$. The
leptokurtic empirical 5-minute distribution leads to the time scaling factor
of $68.05/2.72/96\approx \allowbreak 0.26$ hence the $5$-minute empirical
error is predicted to be 
$$\sqrt{0.26\times 0.25+\sqrt{0.26}(1-0.25)}\approx 67\%$$ of the daily error. 
The actual figure in Table \ref{Table2}
is $0.33/0.548\approx \allowbreak 60\%$. For practical purposes this is again an acceptable approximation.

\begin{table}[tbp] \centering%
\begin{tabular}{lllll}
\hline
& \multicolumn{2}{l}{Black-Scholes} & \multicolumn{2}{l}{empirical L\'{e}vy}
\\ 
$\triangle $ & $\hat{V}_{0}$ & $\hat{\varepsilon}_{0}$ & $V_{0}$ & $%
\varepsilon _{0}$ \\ \hline
5 min & 0.2525 & 0.142 & 0.2548 & 0.330 \\ 
15 min & 0.2535 & 0.191 & 0.2556 & 0.345 \\ 
30 min & 0.2535 & 0.231 & 0.2558 & 0.360 \\ 
1 hr & 0.2536 & 0.282 & 0.2559 & 0.380 \\ 
2 hr & 0.2537 & 0.345 & 0.2560 & 0.421 \\ 
4 hr & 0.2537 & 0.424 & 0.2560 & 0.470 \\ 
8 hr & 0.2538 & 0.519 & 0.2561 & 0.548 \\ \hline
\end{tabular}%
\caption{Mean value $V_0$ and unconditional standard deviation of the hedging error $\varepsilon_0$ 
for parameter values $T=1, S_0=100, B=107.9,K=103.3$}\label{Table2}%
\end{table}%

Our exploratory analysis above points to two open questions in this area of
research: 1) calculation of explicit asymptotic expression for hedging error
of barrier options in discretely rebalanced Black-Scholes model analogous to
the formula of \cite{gobet.temam.01} for path-independent options; 2)
asymptotic formula for hedging error of barrier options in a continuously
rebalanced L\'{e}vy model with small jumps. There is a good reason to
believe that 1)\ and 2) are closely linked because similar link has already
been established for plain vanilla options, see \cite{cerny.al.13}.

\section{Sharpe ratio price bounds}

In this model, as in reality, the sale of an option and subsequent hedging
is a risky activity. If one sells an option at its Black-Scholes value
corresponding to historical volatility one effectively enters into an
investment with zero mean and non-zero variance. In addition this investment
is by construction uncorrelated with the stock returns. To make option
trading profitable the trader must aim for a certain level of risk-adjusted
returns, which implies selling derivatives above their Black--Scholes value.
The question then arises as to what is a sensible measure of risk-adjusted
returns and what is a sensible level of compensation for the residual risk.

\cite{cerny.book.09} proposes to measure profitability of investment by its
certainty equivalent growth rate adjusted for investor's risk aversion. When
this measure is applied to mean-variance preferences, it yields a one-to-one
relationship with the ex-ante Sharpe ratio of the investment strategy. Thus,
in the present context, the unconditional Sharpe ratio appears as a natural
measure of risk-adjusted returns.

It is well known that the square of maximal Sharpe ratio available by
trading in two uncorrelated assets equals the sum of squared Sharpe ratios
of the individual assets. Since the hedged option position is uncorrelated
with the stock we can regard the Sharpe ratio of the hedged position as a
meaningful measure of \emph{incremental} performance (i.e. performance over
and above optimal investment in the stock).

Suppose that the trader targets a certain level of annualized incremental
Sharpe ratio $h$ (say $h=0.5$). Assuming that he or she can sell the option
at price $\tilde{C}$ above the mean value $V_{0}$ the resulting Sharpe ratio
of the hedged option position equals%
\begin{equation*}
\frac{e^{rT}(\tilde{C}-V_{0})}{\varepsilon _{0}}.
\end{equation*}%
If $T$ is maturity in years the trader should look for a price $\tilde{C}$
such that%
\begin{equation*}
\frac{e^{rT}(\tilde{C}-V_{0})}{\varepsilon _{0}}=h\sqrt{T},
\end{equation*}%
which yields%
\begin{equation}
\tilde{C}=V_{0}+e^{-rT}h\sqrt{T}\varepsilon _{0}.
\label{eq: nonlinear pricing}
\end{equation}

For plain vanilla options the price adjustment corresponding to annualized
incremental Sharpe ratio of 1 gives rise to a gap between implied volatility
and historical volatility of about 150 basis points, robustly across
maturities and strikes. If the same price adjustment is performed for
barrier options its magnitude is as important as, and often several times
dominates, the price adjustment due to discrete monitoring. The fraction $%
\frac{\sqrt{T}\varepsilon _{0}}{V_{0}}$ is reported in Tables \ref{Table7} and \ref{Table8}.

\begin{table}[htbp] \centering%
\begin{tabular}{rrcrrrrr}
\hline
      &       & \multicolumn{6}{c}{\textbf{barrier (delta/level)}} \bigstrut\\
\cline{3-8}\multicolumn{2}{c}{strike} & 1E-100 & \multicolumn{1}{c}{0.01} & \multicolumn{1}{c}{0.10} & \multicolumn{1}{c}{0.30} & \multicolumn{1}{c}{0.45} & \multicolumn{1}{c}{0.49} \bigstrut[t]\\
\multicolumn{1}{c}{delta} & \multicolumn{1}{c}{level} & 343.8 & \multicolumn{1}{c}{114.6} & \multicolumn{1}{c}{107.9} & \multicolumn{1}{c}{103.3} & \multicolumn{1}{c}{100.9} & \multicolumn{1}{c}{100.3} \bigstrut[b]\\
\hline
\multicolumn{1}{c}{0.01} & \multicolumn{1}{c}{114.6} & \textbf{177\%} &       &       &       &       &  \bigstrut[t]\\
\multicolumn{1}{c}{0.10} & \multicolumn{1}{c}{107.9} & \textbf{35\%} & \multicolumn{1}{c}{\textbf{75\%}} &       &       &       &  \\
\multicolumn{1}{c}{0.30} & \multicolumn{1}{c}{103.3} & \textbf{13\%} & \multicolumn{1}{c}{\textbf{26\%}} & \multicolumn{1}{c}{\textbf{61\%}} &       &       &  \\
\multicolumn{1}{c}{0.45} & \multicolumn{1}{c}{100.9} & \textbf{7\%} & \multicolumn{1}{c}{\textbf{16\%}} & \multicolumn{1}{c}{\textbf{35\%}} & \multicolumn{1}{c}{\textbf{121\%}} &       &  \\
\multicolumn{1}{c}{0.49} & \multicolumn{1}{c}{100.3} & \textbf{7\%} & \multicolumn{1}{c}{\textbf{15\%}} & \multicolumn{1}{c}{\textbf{30\%}} & \multicolumn{1}{c}{\textbf{94\%}} & \multicolumn{1}{c}{\textbf{491\%}} &  \\
\multicolumn{1}{c}{0.75} & \multicolumn{1}{c}{96.3} & \textbf{3\%} & \multicolumn{1}{c}{\textbf{8\%}} & \multicolumn{1}{c}{\textbf{17\%}} & \multicolumn{1}{c}{\textbf{34\%}} & \multicolumn{1}{c}{\textbf{75\%}} & \multicolumn{1}{c}{\textbf{104\%}} \\
\multicolumn{1}{c}{0.99} & \multicolumn{1}{c}{87.5} & \textbf{0.2\%} & \multicolumn{1}{c}{\textbf{4\%}} & \multicolumn{1}{c}{\textbf{9\%}} & \multicolumn{1}{c}{\textbf{15\%}} & \multicolumn{1}{c}{\textbf{26\%}} & \multicolumn{1}{c}{\textbf{33\%}} \bigstrut[b]\\
\hline
\end{tabular}%
\caption{Risk premium as a percentage of mean value for up-and-out call. $T=$ 1 month, 
$\Delta =$ 1 day, $\mu=0.1,  r=0$. Strike and barrier levels are parametrized by the Black--Scholes 
delta of their position.}\label{Table7}%
\end{table}%

\begin{table}[htbp] \centering%
\begin{tabular}{rrcrrrrr}
\hline
      &       & \multicolumn{6}{c}{\textbf{barrier (delta/level)}} \bigstrut\\
\cline{3-8}\multicolumn{2}{c}{strike} & 1E-100 & \multicolumn{1}{c}{0.01} & \multicolumn{1}{c}{0.10} & \multicolumn{1}{c}{0.30} & \multicolumn{1}{c}{0.45} & \multicolumn{1}{c}{0.49} \bigstrut[t]\\
\multicolumn{1}{c}{delta} & \multicolumn{1}{c}{level} & 2071.0 & \multicolumn{1}{c}{140.5} & \multicolumn{1}{c}{121.2} & \multicolumn{1}{c}{108.8} & \multicolumn{1}{c}{102.8} & \multicolumn{1}{c}{101.4} \bigstrut[b]\\
\hline
\multicolumn{1}{c}{0.01} & \multicolumn{1}{c}{114.6} & \textbf{111\%} &       &       &       &       &  \bigstrut[t]\\
\multicolumn{1}{c}{0.10} & \multicolumn{1}{c}{107.9} & \textbf{17\%} & \multicolumn{1}{c}{\textbf{55\%}} &       &       &       &  \\
\multicolumn{1}{c}{0.30} & \multicolumn{1}{c}{103.3} & \textbf{6\%} & \multicolumn{1}{c}{\textbf{19\%}} & \multicolumn{1}{c}{\textbf{43\%}} &       &       &  \\
\multicolumn{1}{c}{0.45} & \multicolumn{1}{c}{100.9} & \textbf{3\%} & \multicolumn{1}{c}{\textbf{11\%}} & \multicolumn{1}{c}{\textbf{24\%}} & \multicolumn{1}{c}{\textbf{103\%}} &       &  \\
\multicolumn{1}{c}{0.49} & \multicolumn{1}{c}{100.3} & \textbf{3\%} & \multicolumn{1}{c}{\textbf{11\%}} & \multicolumn{1}{c}{\textbf{19\%}} & \multicolumn{1}{c}{\textbf{69\%}} & \multicolumn{1}{c}{\textbf{665\%}} &  \\
\multicolumn{1}{c}{0.75} & \multicolumn{1}{c}{96.3} & \textbf{1\%} & \multicolumn{1}{c}{\textbf{6\%}} & \multicolumn{1}{c}{\textbf{11\%}} & \multicolumn{1}{c}{\textbf{21\%}} & \multicolumn{1}{c}{\textbf{47\%}} & \multicolumn{1}{c}{\textbf{71\%}} \\
\multicolumn{1}{c}{0.99} & \multicolumn{1}{c}{87.5} & \textbf{0\%} & \multicolumn{1}{c}{\textbf{3\%}} & \multicolumn{1}{c}{\textbf{5\%}} & \multicolumn{1}{c}{\textbf{8\%}} & \multicolumn{1}{c}{\textbf{14\%}} & \multicolumn{1}{c}{\textbf{19\%}} \bigstrut[b]\\
\hline
\end{tabular}%
\caption{Risk premium as a percentage of mean value for up-and-out call. $T=$ 6 months, 
$\Delta =$ 1 day, $\mu=0.1,  r=0$. Strike and barrier levels are parametrized by the Black--Scholes 
delta of their position.}\label{Table8}%
\end{table}%

\section{Conclusions}

One obvious conclusion to draw from formula (\ref{eq: nonlinear pricing}) is
that prices in an incomplete market are likely to contain both a linear $%
(V_{0})$ and a non-linear $(\varepsilon _{0})$ component. The prevailing
market practice is to use just the linear part $V_{0}$ for calibration which
often requires distorting the historical distribution of returns to match
observed market prices across strikes and maturities. For example, in their
calibration of plain vanilla option prices \cite{Madetal99} report
historical annualized excess kurtosis at 0.002 but risk-neutral excess
kurtosis at 0.18 which is a level that the variance-optimal martingale
measure that generates $V_{0}$ simply cannot reach. This phenomenon gets
worse in the presence of exotic options. Formula (\ref{eq: nonlinear pricing}%
) offers a flexible alternative that may offer better fit of model dynamics
to historical return distributions and at the same time provide closer
calibration to market prices thanks to the non-linear term $\varepsilon _{0}$
which has very different characteristics for different types of exotic
options, as we have seen in the previous section.

\def\cprime{$'$}

\end{document}